\documentclass[10pt,twocolumn,twoside]{IEEEtran}
\IEEEoverridecommandlockouts 

\usepackage{cite}
\usepackage{amsmath,amssymb,amsfonts,amsthm}
\usepackage{algorithmic}
\usepackage{graphicx}
\usepackage{textcomp}
\usepackage{balance}
\usepackage{xcolor}
\usepackage[utf8]{inputenc}
\usepackage[normalem]{ulem}
\usepackage{enumerate}
\usepackage{subcaption}
\usepackage{booktabs} 
\usepackage{multirow} 
\usepackage{mathtools}
\mathtoolsset{showonlyrefs=true}
\usepackage{tikz}
\usetikzlibrary{shapes.geometric}

\newcommand{\filledstar}{%
  \tikz[baseline=-0.55ex]
  \node[
    star,
    star points=5,
    star point ratio=2.25,
    fill=black,
    draw=black,
    inner sep=0pt,
    minimum size=0.9em
  ] {};
}

\newcommand{\hollowstar}{%
  \tikz[baseline=-0.55ex]
  \node[
    star,
    star points=5,
    star point ratio=2.25,
    draw=black,
    fill=white,
    inner sep=0pt,
    minimum size=0.9em
  ] {};
}

\def\BibTeX{{\rm B\kern-.05em{\sc i\kern-.025em b}\kern-.08em
    T\kern-.1667em\lower.7ex\hbox{E}\kern-.125emX}}
\usepackage{lipsum}

\newtheorem{theorem}{Theorem}
\newtheorem{lemma}{Lemma}

\newtheorem{definition}{Definition}
\newtheorem{remark}{Remark}
\newcommand{\rac}[1]{\textcolor{black}{#1}}

\newcommand{\diag}{\mathrm{diag}}

\title{\LARGE \bf 
Decentralized Stability Certificates in IBR-Dominated Grids: \\The Role of the Network State
} 

\author{Zhimeng Wang$^\dagger$, Sushobhan Chatterjee$^\dagger$, Sijia Geng$^\dagger$, Richard Pates$^\ddagger$, Enrique Mallada$^\dagger$
\vspace{-0.5cm}
\thanks{ 
$^\dagger$Zhimeng Wang, Sushobhan Chatterjee, Sijia Geng, Enrique Mallada are with the Department of Electrical and Computer Engineering, Johns Hopkins University, Baltimore, MD, USA. Email: {\tt \{zwang471, schatt21, sgeng, mallada\}@jhu.edu}}
\thanks{$^\ddagger$Richard Pates is with the Department of Automatic Control, Lund University, Lund, Sweden. Email: {\tt \{richard.pates\}@control.lth.se}
}
\thanks{This work was supported by the NSF Global Centers program under Grant No.~2330450 and by the DOE Office of Science (ASCR) under Award No.~826565.}
}

\usepackage{ifthen}
\newboolean{showcomments}
\setboolean{showcomments}{true}
\usepackage{color}
\usepackage{todonotes}

\definecolor{bleudefrance}{rgb}{0.19, 0.55, 0.91}
\definecolor{ao(english)}{rgb}{0.0, 0.5, 0.0}

\newcommand{\addcite}[0]{\ifthenelse{\boolean{showcomments}}
{\textcolor{purple}{(add cite(s)) }}{}}%

\newcommand{\addref}[0]{\ifthenelse{\boolean{showcomments}}
{\textcolor{purple}{(add ref) }}{}}%

\newcommand{\enrique}[1]{  \ifthenelse{\boolean{showcomments}}
{\todo[inline,caption={},color=bleudefrance]{Enrique: #1}}{}}
\newcommand{\rene}[1]{  \ifthenelse{\boolean{showcomments}}
{\todo[inline,caption={},color=cyan]{Ren\'e: #1}}{}}
\newcommand{\emmargin}[1]{\ifthenelse{\boolean{showcomments}}{\marginpar{\color{bleudefrance}\tiny EM: #1}}{}}
\newcommand{\hancheng}[1]{  \ifthenelse{\boolean{showcomments}}
{\todo[inline,color=orange]{Hancheng: #1}}{}}
\newcommand{\ziqing}[1]{  \ifthenelse{\boolean{showcomments}}
{\todo[inline,color=red]{Ziqing: #1}}{}}
\newcommand{\salma}[1]{  \ifthenelse{\boolean{showcomments}}
{\todo[inline,color=yellow]{Salma: #1}}{}}
\newcommand{\zhimeng}[1]{  \ifthenelse{\boolean{showcomments}}
{\todo[inline,color=teal]{Zhimeng: #1}}{}}
\newcommand{\zxmargin}[1]{\ifthenelse{\boolean{showcomments}}{\marginpar{\color{purple}\tiny ZX: #1}}{}}
\newcommand{\stmargin}[1]{\ifthenelse{\boolean{showcomments}}{\marginpar{\color{red}\tiny ST: #1}}{}}
\newcommand{\sushobhan}[1]{  \ifthenelse{\boolean{showcomments}}
{\todo[inline,color=pink]{Sushobhan: #1}}{}}

\newcommand{\hl}[1]{\ifthenelse{\boolean{showcomments}}
{\textcolor{red}{#1}}{#1}}

\newboolean{showedits}
\setboolean{showedits}{true}
\usepackage[markup=underlined]{changes}
\definechangesauthor[color=red]{EM}
\newcommand{\aem}[1]{
\ifthenelse{\boolean{showedits}}
{\added[id=EM]{#1}}
{\!#1\hspace{-4.75pt}}
}
\newcommand{\repem}[2]{
\ifthenelse{\boolean{showedits}}
{\replaced[id=EM]{#1}{#2}}
{\!#1\hspace{-4.75pt}}
}
\newcommand{\dem}[1]{
\ifthenelse{\boolean{showedits}}
{\deleted[id=EM]{#1}}
{}
}

\begin{document}
\begingroup
\allowdisplaybreaks

\maketitle

\begin{abstract}
Small-signal instabilities, such as unforced sub-synchronous oscillations (SSOs), are increasingly observed in inverter-based resource (IBR) dominated grids.
While decentralized stability certificates offer a scalable means to avoid instability onset, they are typically derived under restrictive network-state assumptions--such as small angle differences or negligible voltage drops--that cannot capture how departures from these conditions affect system stability.
In this paper, we develop a network model and a decentralized analysis framework that explicitly characterizes how reactive power mismatches, line loading, and inverter control parameters jointly determine small-signal stability.
We show that increased steady-state reactive power mismatches and line loading lead to more stringent conditions on admissible inverter droop gains.
These results make decentralized stability certificates explicitly network-state dependent, showing how network stress shrinks the set of stabilizing local controller parameters.
\end{abstract}


\section{Introduction}

Small-signal instabilities are becoming increasingly prevalent in inverter-based resource (IBR) dominated grids, often manifesting as unforced sub-synchronous oscillations (SSOs) \cite{fan2022real}. Such instabilities have been reported across a wide range of scenarios, including weak-grid conditions, controller misconfigurations at commissioning, and post-contingency network reconfigurations \cite{li2019wind,damas2020subsynchronous}. 
While some events can be explained by classical resonance phenomena or control-structure interactions \cite{shair2019overview}, the precise mechanisms driving many of these oscillations remain poorly understood. In particular, instabilities often arise spontaneously, without a clear disturbance or topological trigger, underscoring the need for improved analytical tools that capture how the underlying network state shapes small-signal dynamics and stability-margins~\cite{cheng2022real}.

Recent research has emphasized the use of \emph{decentralized certificates} to ensure small signal stability~\cite{pates2019robust,vorobev2019decentralized,watson2020control,siahaan2024decentralized,huang2024gain,haberle2025decentralized}. These approaches derive scale-free conditions that guarantee stability under bounded uncertainty, offering scalable alternatives to full electromagnetic transient (EMT) simulations and global eigenvalue analyses\cite{pates2019robust}. To obtain such tractable, decentralized conditions, most formulations either adopt simplifying network steady-state assumptions--such as small phase differences~\cite{vorobev2019decentralized,haberle2025decentralized}, near-uniform voltage magnitudes~\cite{watson2020control}, or active–reactive decoupling~\cite{pates2019robust,siahaan2024decentralized}--or implicitly encode network-state dependence within the local linearization of inverter models~\cite{huang2024gain}. As a result, the dependence of stability margins on the network state is either neglected, implicitly confined to a nominal regime, or absorbed into model abstractions.

In this paper, we develop a network model and accompanying analysis framework that derives decentralized stability certificates without imposing explicit or implicit assumptions on the network state. \textcolor{black}{The port-based representation and loop-transformation perspective used in our framework are inspired by the scalable decentralized stability framework in \cite{pates2014scalable}, which develops local stability protocols for heterogeneous networks. }Our approach relies on a small-signal model in which lines and IBRs are represented as multi-port components, interconnected through an incidence-matrix-based network structure. In particular, each transmission line is modeled as a quasi-stationary multi-port subsystem mapping small deviations of nodal angles and voltages to directional real and reactive power flows, while inverter dynamics are represented through filtered droop gains. This explicit port-based decomposition makes the passivity properties of individual components--and potential violations thereof--transparent, enabling a systematic and decentralized application of loop transformation techniques to compensate for non-passive effects. 

Applying this framework reveals an intrinsic coupling between the network steady state and the admissible control gains required for guaranteeing small-signal stability. In particular, our results show that increasing reactive power asymmetries ($Q_{ij}\!\!-\!Q_{ji}$) or line loading--captured through larger phase differences ($\theta_{ij}\!=\!\theta_i\!-\!\theta_j$)--progressively restrict the set of stabilizing voltage droop gains. These findings align with the well-established notion that small-signal stability depends on the operating point of the network, and indicate that decentralized certificates derived under nominal or lightly stressed conditions can be overly optimistic when applied to stressed operating regimes. Moreover, the derived conditions provide practical guidance for preserving stability by indicating when network-aware retuning or redispatch may be required.

\textcolor{black}{Our certificate also relates to conditions that exploit full network information, such as those derived for lossless droop-controlled microgrids in~\cite{schiffer2014conditions} and for dispatchable virtual oscillator control with transmission-line dynamics in~\cite{gross2019effect}. Naturally, a decentralized certificate can only provide a sufficient version of such centralized conditions, trading network-wide quantities for line-wise, locally verifiable tests; ours nonetheless retains a key feature of~\cite{schiffer2014conditions}, namely that only the voltage droop gains are constrained while the frequency droop gains and their filter time constants remain free (Remark~\ref{rem:kp}). Closest to our work is the earlier, independently derived certificate of~\cite{niehues2025small}, which gives bus-local small-signal conditions for heterogeneous devices with a voltage--reactive-power droop. While similar in spirit, our certificate yields a different condition, in which line loading and the directional reactive power flow $Q_{ij}\! -\!\,Q_{ji}$ appear explicitly and the voltage-droop budget of each bus is allocated across its lines. A systematic comparison is left for future work.} 


The rest of the paper is organized as follows. Section~\ref{sec:preliminaries} reviews the preliminary tools used throughout the paper, including the passivity stability theorem and loop transformation techniques for analyzing feedback interconnections. Section~\ref{sec:network_model} introduces the network and device models and describes their interconnection through a structured feedback representation. Section~\ref{sec:certificate} presents the main decentralized stability certificate and provides a detailed proof of the result. Section~\ref{sec:num_eg} demonstrates the implications of the proposed conditions through an illustrative numerical example. Finally, Section~\ref{sec:conclusion} summarizes the main findings and discusses their implications for decentralized stability analysis and inverter-based grid operation.

\section{Preliminaries} \label{sec:preliminaries}
Our analysis relies on passivity theory and standard mechanisms for compensating the lack of passivity in interconnected systems. Here, we introduce passivity stability theorem, which provides sufficient conditions for the stability of feedback interconnections, and introduce loop transformation techniques that allow passivity properties to be modified through equivalent interconnection representations. These tools form the basis for the decentralized stability analysis developed in the remainder of the paper.

\subsection{Internal Feedback Stability}

We consider the standard negative feedback interconnection of two square, proper transfer matrices $H_1(s)$ and $H_2(s)$. The resulting closed-loop input--output relation can be written as
\begin{equation}\label{eq:gang_of_four}
\begin{bmatrix}
y_1 \\ y_2
\end{bmatrix}
\!=\!
\underbrace{
\begin{bmatrix}
(I + H_1 H_2)^{-1} H_1 \!\!&\!\! -(I + H_1 H_2)^{-1} H_1 H_2 \\
H_2 (I + H_1 H_2)^{-1} H_1 \!\!&\!\! H_2 (I + H_1 H_2)^{-1}
\end{bmatrix}
}_{=: \, H_1 \# H_2(s)}
\begin{bmatrix}
u_1 \\ u_2
\end{bmatrix}, 
\end{equation}
where we use $I$ as an identity matrix with proper dimensions, and $I_m$ as an $m\times m$ identity matrix. 

\begin{definition}[Internal Feedback Stability]\label{def:internal_stability}
The negative feedback interconnection of $H_1(s)$, $H_2(s)$ is said to be \emph{internally stable} if all transfer matrices comprising $H_1 \# H_2(s)$ are stable, i.e., all poles of $H_1 \# H_2(s)$ lie in the open left-half plane.
\end{definition}

Under the assumption that there are no right-half-plane pole--zero cancellations between $H_1(s)$ and $H_2(s)$—i.e., all right-half-plane poles of $H_1(s)$ and $H_2(s)$ are contained in the minimal realizations of $H_1(s)H_2(s)$ and $H_2(s)H_1(s)$—internal stability of one closed-loop transfer matrix implies internal stability of all others \cite{skogestad2005multivariable}. This notion will be used throughout the paper.

\subsection{Passivity Stability Theorem}
We briefly review standard passivity definitions and a classical passivity-based stability theorem for feedback interconnections.
\begin{definition}[Passivity]\label{def:passivity_pr}
    A square, proper, real rational transfer matrix $H(s)$ is \emph{passive} if:
    \begin{enumerate}[i.]
        \item $H(s)$ has poles in $\Re(s)\leq0$.
        \item The Hermitian part of $H(j\omega)$ is positive semidefinite for any $\omega$ for which $j\omega$ is not a pole:
        \begin{equation}
            H(j\omega) + H^*(j\omega) \succeq 0.
        \end{equation}
        \item Any pole on the imaginary axis is simple, and its residue
        \begin{equation}
            R = \lim_{s \to j\omega_0} (s - j\omega_0) H(s)
        \end{equation}
        is Hermitian positive semidefinite.
        
    \end{enumerate}
The transfer matrix is \emph{strictly passive} if inequalities in $i.$ and $ii.$ are strict.
\end{definition}

\begin{theorem}[Feedback Stability via Passivity]\label{thm:passivity_stability}
    Consider the standard negative feedback interconnection of two systems $H_1(s)$ and $H_2(s)$ with compatible dimensions. The closed-loop system is internally stable if the following conditions hold:
    \begin{enumerate}[i.]
        \item System $H_1(s)$ is passive, and system $H_2(s)$ is strictly passive.
        \item The loop gain magnitude at infinite frequency satisfies the small-gain condition:
        \begin{equation}
            \bar{\sigma}\left( H_1(j\infty) \right) \bar{\sigma}\left( H_2(j\infty) \right) < 1.
        \end{equation}
    \end{enumerate}
\end{theorem}
\begin{proof}
    The proof of this Theorem is standard. 
    For completeness we provide a proof in Appendix \ref{appendix:thm1}.
\end{proof}

\begin{figure}

\begin{subfigure}{0.25\textwidth}
\includegraphics[width=0.9\linewidth]{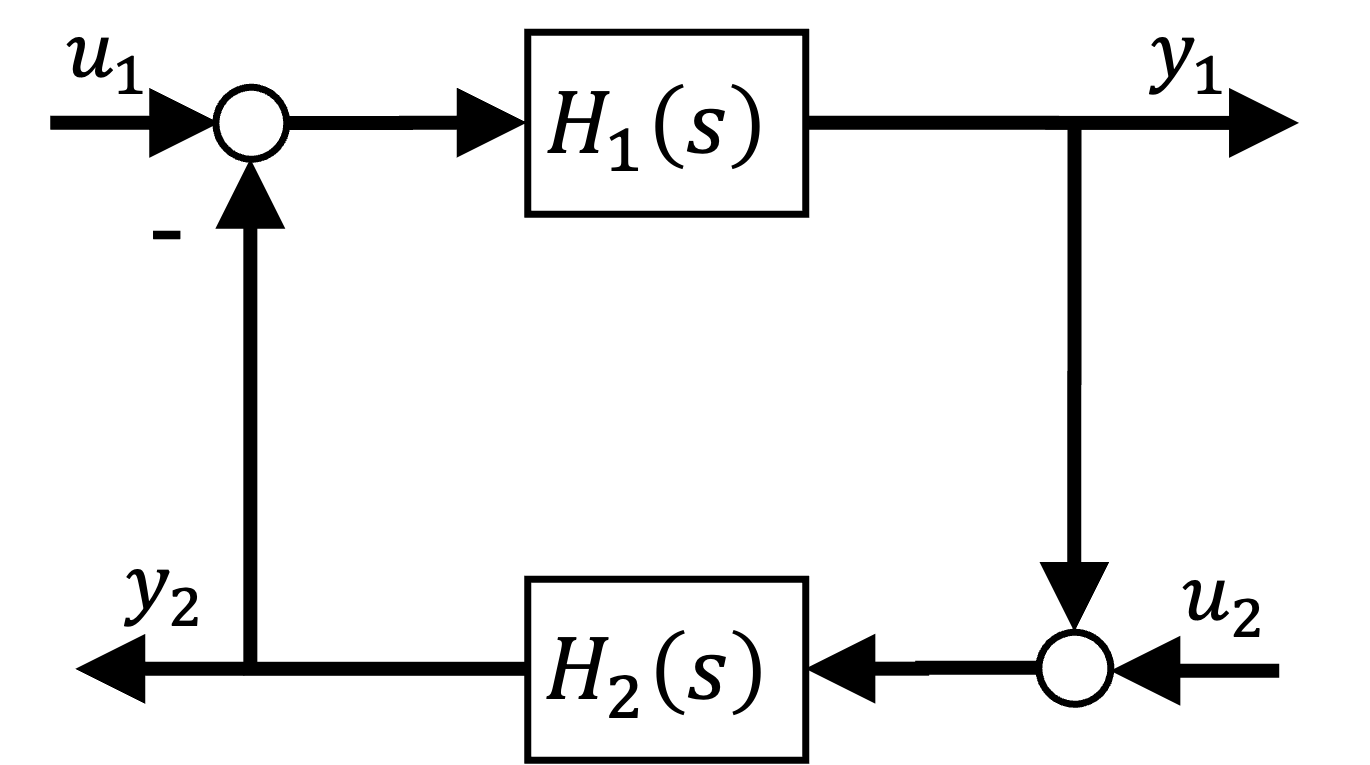}
\caption{Negative feedback interconnection of $H_1(s)$ and $H_2(s)$. } \label{fig:before_loop_transformation_012026}
\end{subfigure}\hspace*{\fill}
\begin{subfigure}{0.25\textwidth}
\includegraphics[width=0.9\linewidth]{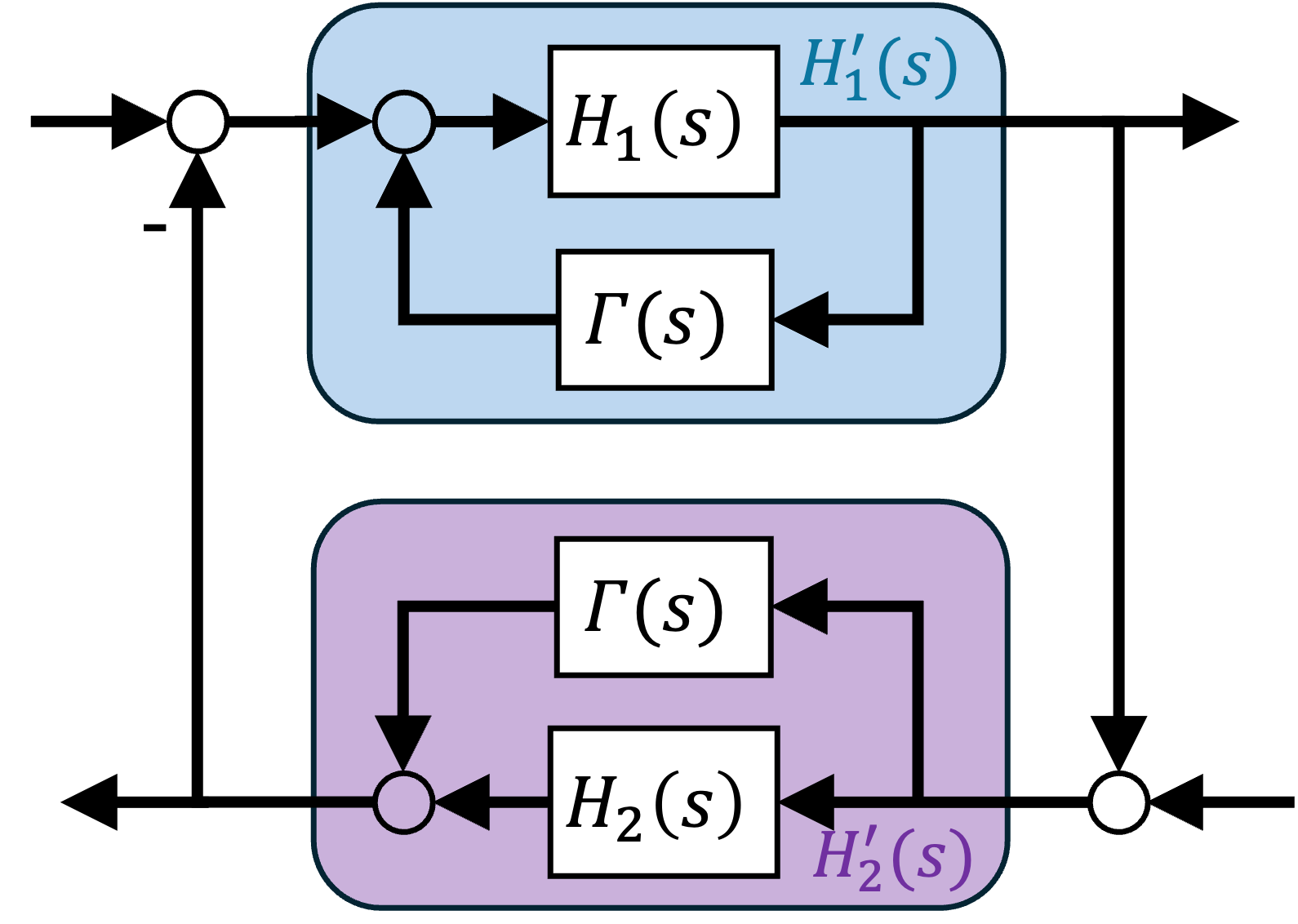}
\caption{Fig.~\ref{fig:before_loop_transformation_012026} with loop transformation $\Gamma(s)$. } \label{fig:after_loop_transformation_062226}
\end{subfigure}
\caption{Loop transformation. } 

\label{fig:loop_transformation}
\end{figure}

\subsection{Loop Transformation}
We will use loop transformations to compensate the lack of passivity in transmission lines, allowing in this way to apply Theorem \ref{thm:passivity_stability}. 
When doing this, one needs to ensure that the internal stability of the loop-shifted system implies the internal stability of the original system. 
We address this issue with the following two lemmas. 

\begin{lemma}[Corollary 5.5 in \cite{zhou1996robust}]\label{lemma:zhou_55}
    Suppose $H_1(s)$ is stable. Then system $H_1(s) \# H_2(s)$ is internally stable if and only if $\left[I+H_2(s) H_1(s)\right]^{-1}H_2(s)$ is stable. 
\end{lemma}

\begin{lemma}[Loop transformation internal stability equivalence]\label{lem:loop_transformation_equiv}
    Let $H_1(s)$ and $H_2(s)$ be proper real-rational transfer matrices with compatible dimensions, and let $\Gamma(s)$ be a stable proper transfer matrix of compatible dimension. 
    Suppose that \[ \det\!\left(I-\Gamma(s)H_1(s)\right)\not\equiv 0, \]
    and define
\[
H_1'(s) := H_1(s)\left(I \!-\! \Gamma(s)H_1(s)\right)^{-1}, \enspace H_2'(s) := H_2(s) + \Gamma(s).
\]
Assume that $H_1'(s)$ is proper real-rational and stable. Assume also that $H_1(s)$ is stable and that 
\[ \det\!\left(I+H_2(\infty)H_1(\infty)\right)\neq 0, \quad \det\!\left(I+H_2'(\infty)H_1'(\infty)\right)\neq 0. \]
Then $H_1(s) \# H_2(s)$ is internally stable if and only if $H_1'(s) \# H_2'(s)$ is internally stable.
\end{lemma}
\begin{proof}
    See Appendix \ref{appendix:cor21}. 
\end{proof}
\begin{remark}[Necessity of Lemma~\ref{lem:loop_transformation_equiv}]
    A classical approach to address the equivalence of internal stability of loop-transformed systems usually requires no pole-zero cancellations on the \emph{closed} right-half plane \cite{doyle2013feedback}. 
    Unfortunately, the proposed system in this paper has a pole-zero cancellation on the imaginary axis, rendering this approach not applicable. 
    We make an emphasis on the requirement of \emph{closed} as some literature rather state it as \emph{open} \cite{skogestad2005multivariable} \cite{zhou1996robust}. 
\end{remark}


\section{Network Model} \label{sec:network_model}
In this section, we develop a first principle-based small-signal model of the IBRs and the transmission network, and then connect them to construct an interconnection. 

\subsection{Inverter-Based Resource (IBR) Model}

We model each IBR as a grid-forming (GFM) device that regulates frequency through active power feedback and voltage magnitude through reactive power feedback. 
Let $\Delta P_i$ and $\Delta Q_i$ denote the incremental active and reactive power injected by inverter $i$ into the network, and let $\Delta \omega_{i}$ and $\Delta \ln |v_i|$ denote the incremental frequency and logarithmic voltage magnitude at bus $i$. 
\textcolor{black}{Logarithmic voltage variables have also appeared in other works on analyses of droop-controlled inverter networks; see, e.g., \cite{schiffer2014conditions}.}

We adopt the standard sign convention in which positive power corresponds to injection into the grid, and the inverter control uses the \emph{negative} of these quantities as feedback inputs. 
\textcolor{black}{We assume the devices to be GFM devices with filter-droop control while neglecting fast inner-loop dynamics.} 
The small-signal inverter dynamics are \textcolor{black}{therefore} modeled by
\begin{equation}\label{eq:Gi_def}
\begin{bmatrix}
\Delta \omega_i\\
\Delta \ln|v_i|
\end{bmatrix}
=
G_i(s)
\begin{bmatrix}
-\Delta P_i\\
-\Delta Q_i
\end{bmatrix},
\end{equation}
where the transfer matrix
\begin{equation}\label{eq:Gi_tf}
G_i(s)
:=
\begin{bmatrix}
\frac{k_{p,i}}{\tau_{\omega,i}s+1} & 0\\
0 & \frac{k_{q,i}}{\tau_{v,i}s+1}
\end{bmatrix}, 
\end{equation}
$k_{p,i}, k_{q,i}>0$ are the active power–frequency droop gain and reactive power-voltage droop gain, respectively. 
$\tau_{\omega,i}, \tau_{v,i}>0$ are the frequency and voltage control time constants, respectively. 

We define
\[
\begin{bmatrix}
    -\Delta P \\
    -\Delta Q
\end{bmatrix} \!:=\! \begin{bmatrix}
    -\Delta P_{1} \\
    -\Delta Q_{1} \\
    \vdots \\
    -\Delta P_{n} \\
    -\Delta Q_{n}
\end{bmatrix} \!\in\mathbb{R}^{2n}, 
\begin{bmatrix}
    \Delta \omega \\
    \Delta \ln |v|
\end{bmatrix} \!:=\! \begin{bmatrix}
    \Delta \omega_{1} \\
    \Delta \ln |v_{1}| \\
    \vdots \\
    \Delta \omega_{n} \\
    \Delta \ln |v_{n}|
\end{bmatrix}\!\in\mathbb{R}^{2n}
\]
as the input and output signals of the IBR devices. 
Thus the dynamics of all the inverters concatenated is characterized by
\begin{equation} \label{eq:ibr_model}
    \begin{bmatrix}
    \Delta \omega \\
    \Delta \ln |v|
\end{bmatrix} = G(s) \begin{bmatrix}
    -\Delta P \\
    -\Delta Q
\end{bmatrix}, 
\end{equation}
where $G(s) = \mathrm{blkdiag}\left(
G_{1}(s),\cdots, G_{n}(s)\right)$. 

\subsection{Line Model}
Consider a lossless transmission line $e=\{i,j\}$ connecting buses $i$ and $j$.
Let
\[
v_i = |v_i|e^{j\theta_i}, \qquad
v_j = |v_j|e^{j\theta_j}, \qquad
\theta_{ij} := \theta_i-\theta_j
\]
denote the steady-state voltage phasors and phase difference.
The line is modeled as purely inductive with inductance $L_e$, and we define the normalized susceptance
$b_e := \frac{1}{\omega_0 L_e}$, where $\omega_0$ is the synchronous angular frequency.

The steady-state complex power injected from bus $i$ into the line is
\[
S_{ij} := P_{ij} + jQ_{ij} = v_i\, i_{ij}^*,
\]
with $i_{ij} = j b_e (v_i - v_j)$, and $S_{ji}$ defined analogously.

Linearizing the AC power-flow equations around the operating point yields a quasi-stationary, port-based small-signal model of the line,
\begin{equation}\label{eq:single_line_model}
\begin{bmatrix}
\Delta P_{ij} \\
\Delta Q_{ij} \\
\Delta P_{ji} \\
\Delta Q_{ji}
\end{bmatrix}
=
N_e(s)
\begin{bmatrix}
\Delta \omega_i \\
\Delta \ln |v_i| \\
\Delta \omega_j \\
\Delta \ln |v_j|
\end{bmatrix},
\end{equation}
where the line transfer matrix admits the factorization
\begin{equation}\label{eq:Ne_factorization}
N_e(s)
=
J_e\,
\mathrm{diag}\!\left(\frac{1}{s},\,1,\,\frac{1}{s},\,1\right).
\end{equation}
We introduce the notation $W_n(s) = \diag\left(\frac{1}{s}, 1, \ldots, \frac{1}{s}, 1\right) \in \mathbb{C}^{n \times n}$, thus $N_e(s) = J_{e}W_4(s)$. 
The integrator terms arise from the kinematic relation $\dot{\theta}_k=\omega_k$, while the static gain matrix
$J_e := b_e |v_i||v_j|\,\bar J_e$ captures the steady-state sensitivity of line power flows to voltage magnitudes and phase differences, and the normalized Jacobian $\bar J_e \in \mathbb{R}^{4\times 4}$ is given by
\[
\bar J_e \!=\!
\begin{bmatrix}
\cos\theta_{ij} \!&\! \sin\theta_{ij} \!&\! -\cos\theta_{ij} \!&\! \sin\theta_{ij}\\
\sin\theta_{ij} \!&\! 2\frac{|v_i|}{|v_j|}-\cos\theta_{ij} \!&\! -\sin\theta_{ij} \!&\! -\cos\theta_{ij}\\
-\cos\theta_{ij} \!&\! -\sin\theta_{ij} \!&\! \cos\theta_{ij} \!&\! -\sin\theta_{ij}\\
\sin\theta_{ij} \!&\! -\cos\theta_{ij} \!&\! -\sin\theta_{ij} \!&\! 2\frac{|v_j|}{|v_i|}-\cos\theta_{ij}
\end{bmatrix}.
\]

Let $E=\{e_1,\dots,e_m\}$ denote the set of transmission lines.
Stacking the incremental nodal variables and line power flows over all lines, we define
\[
\begin{bmatrix}
\Delta \omega_E \\
\Delta \ln |v_E|
\end{bmatrix}
:=
\begin{bmatrix}
\Delta \omega_{i_1} \\
\Delta \ln |v_{i_1}| \\
\Delta \omega_{j_1} \\
\Delta \ln |v_{j_1}| \\
\vdots \\
\Delta \omega_{i_m} \\
\Delta \ln |v_{i_m}| \\
\Delta \omega_{j_m} \\
\Delta \ln |v_{j_m}|
\end{bmatrix},
\qquad
\begin{bmatrix}
\Delta P_E \\
\Delta Q_E
\end{bmatrix}
:=
\begin{bmatrix}
\Delta P_{i_1 j_1} \\
\Delta Q_{i_1 j_1} \\
\Delta P_{j_1 i_1} \\
\Delta Q_{j_1 i_1} \\
\vdots \\
\Delta P_{i_m j_m} \\
\Delta Q_{i_m j_m} \\
\Delta P_{j_m i_m} \\
\Delta Q_{j_m i_m}
\end{bmatrix}.
\]

The concatenated line dynamics can then be written compactly as
\begin{equation}\label{eq:line_model}
\begin{bmatrix}
\Delta P_E \\
\Delta Q_E
\end{bmatrix}
=
N_E(s)
\begin{bmatrix}
\Delta \omega_E \\
\Delta \ln |v_E|
\end{bmatrix},
\end{equation}
where
\[
N_E(s)
=
\mathrm{blkdiag}\bigl(N_{e_1}(s),\dots,N_{e_m}(s)\bigr)
=
J_E\,W_{4m}(s),
\]
with
$J_E := \mathrm{blkdiag}\!\left(J_{e_1},\dots,J_{e_m}\right)$. 

\subsection{Interconnection Model}

\begin{figure*}
\centering
\begin{subfigure}{0.25\textwidth} 
\includegraphics[width=0.9\linewidth]{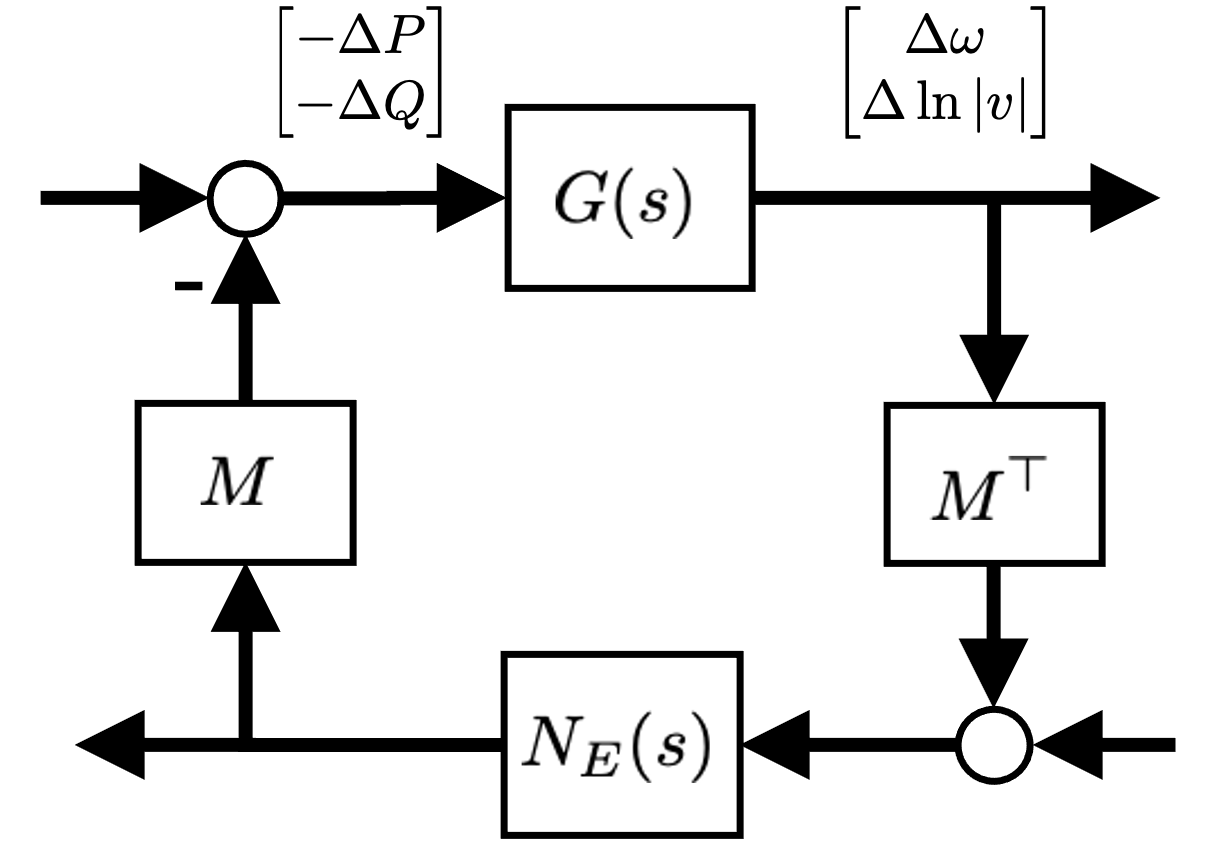}
\caption{System $G(s)\#_M N_E(s)$.} 
\label{fig:our_original_loop}
\end{subfigure}\hspace*{\fill}
\begin{subfigure}{0.25\textwidth} 
\includegraphics[width=0.9\linewidth]{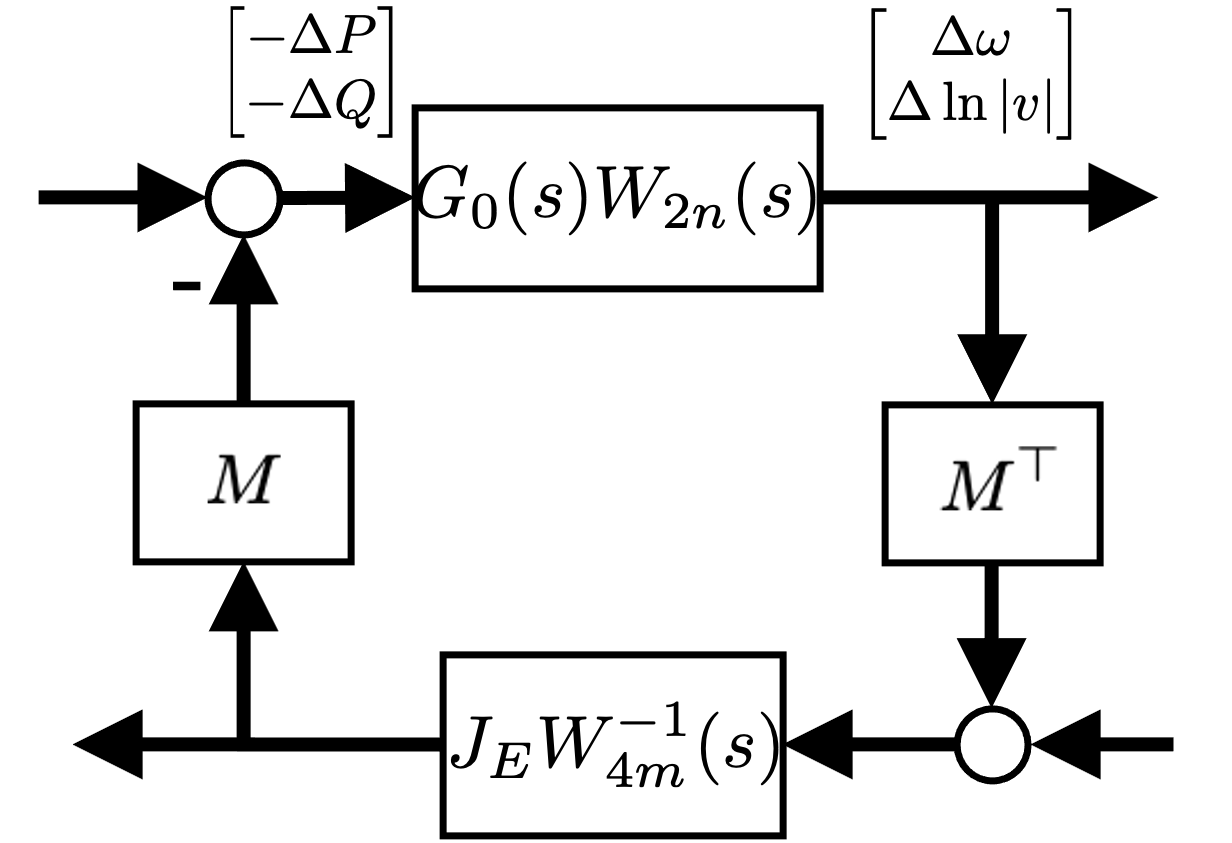}
\caption{ $G_0(s) W_{2n}(s)\#_M J_E W_{4m}^{-1}(s)$.} 
\label{fig:first_step_loop_shifting}
\end{subfigure}\hspace*{\fill}
\begin{subfigure}{0.25\textwidth} 
\includegraphics[width=0.9\linewidth]{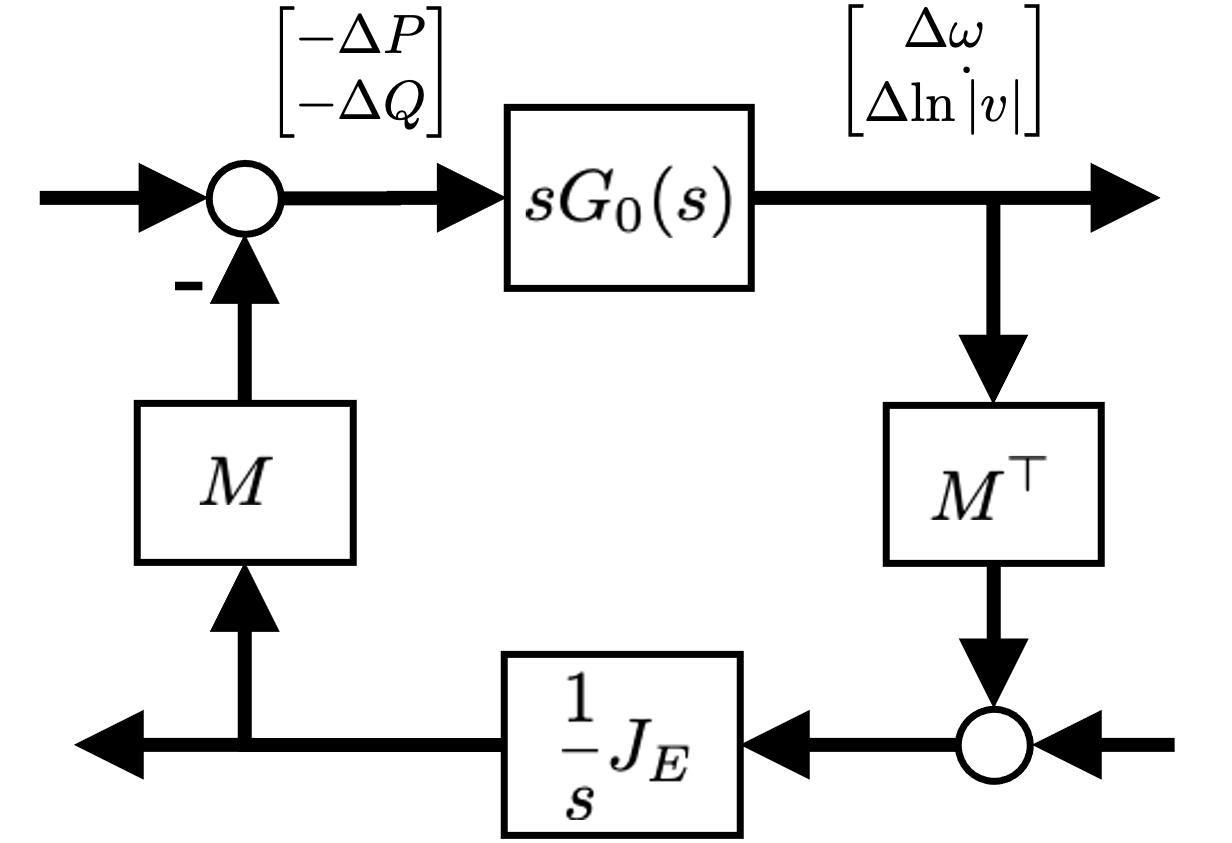}
\caption{System $sG_0(s) \#_M \frac{1}{s}J_E$. } 
\label{fig:before_loop_shifting}
\end{subfigure}\hspace*{\fill}
\begin{subfigure}{0.25\textwidth}
\includegraphics[width=0.9\linewidth]{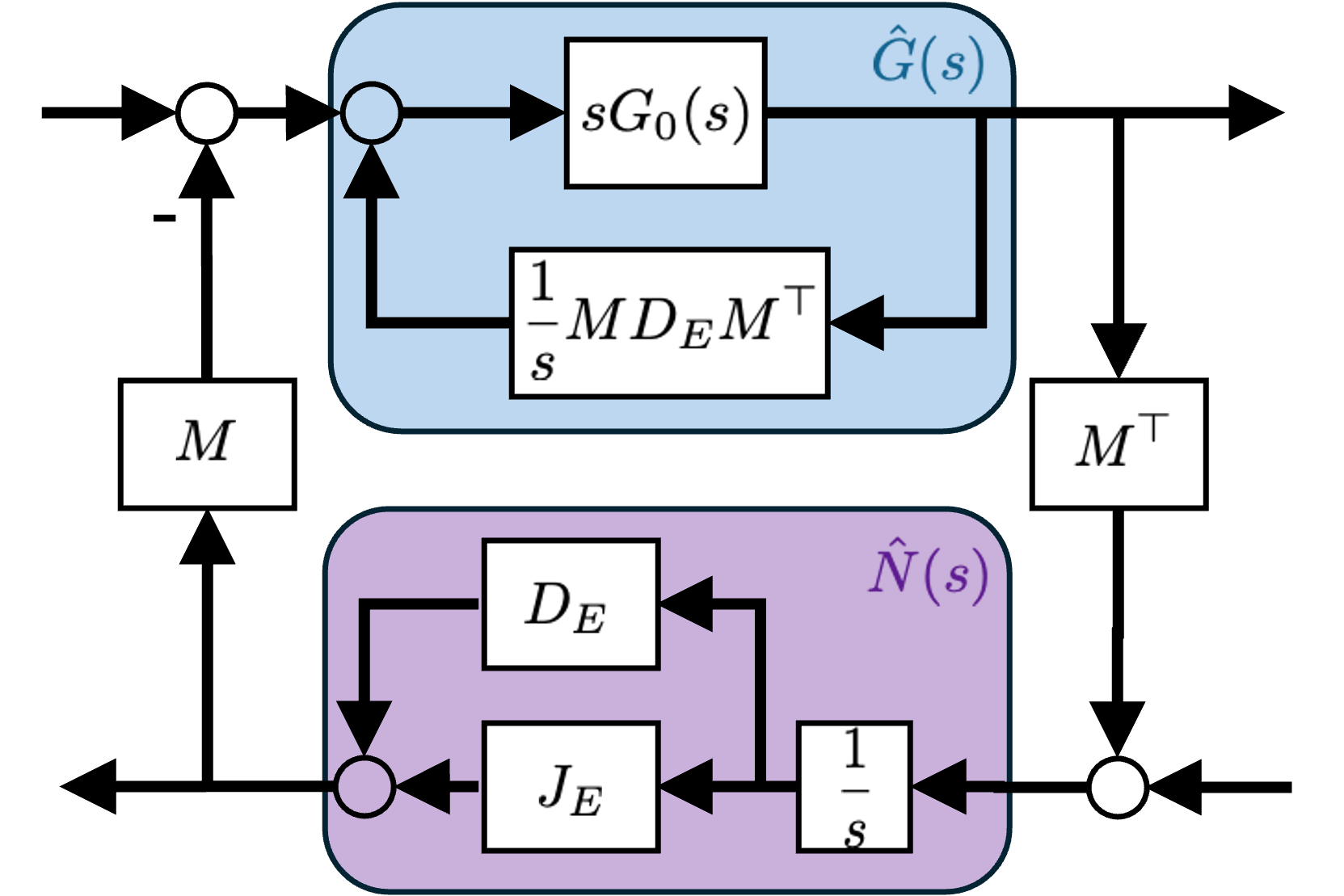}
\caption{System $\hat G(s) \#_M \hat N(s)$. } 
\label{fig:our_loop_shifted_loop_color}
\end{subfigure}
\caption{Systems for stability analysis. } 

\label{fig:our_loop_transformation}
\end{figure*}

To connect the bus vector signals and the line-end vector signals, we introduce the matrix
\begin{equation}\label{eq:def_M}
    M:= \begin{bmatrix}
        M_{e_1} & M_{e_2} & \cdots & M_{e_m}
    \end{bmatrix} \in \{0,1\}^{2n\times 4m}, 
\end{equation}
where for each line, $M_{e}=\begin{bmatrix}
        e_{i} & e_{j}
    \end{bmatrix}\otimes I_{2} \in \{0,1\}^{2n\times 4}$, and $e_k$ is the $k$-th standard basis vector. 

Therefore, the bus vector signals and the line-end vector signals can be related as 
\begin{equation}\label{eq:lifting_and_aggregation}
\begin{bmatrix}
    \Delta \omega_{E} \\
    \Delta \ln |v_{E}|
\end{bmatrix} = M^{\top}\begin{bmatrix}
    \Delta \omega \\
    \Delta \ln |v|
\end{bmatrix}, \quad \begin{bmatrix}
    \Delta P \\
    \Delta Q
\end{bmatrix} = M \begin{bmatrix}
    \Delta P_{E} \\
    \Delta Q_{E}
\end{bmatrix}
\end{equation}

Equations \eqref{eq:ibr_model}, \eqref{eq:line_model}, and \eqref{eq:lifting_and_aggregation} provide a compact first-principles description of the full networked system, in which the inverter dynamics and the line dynamics are interconnected exclusively through the operators $M^\top$ (node-to-edge) and $M$ (edge-to-node). 
\textcolor{black}{This port-based representation is inspired by the network decomposition framework in \cite{pates2014scalable}, where heterogeneous components are represented through local input-output ports and coupled through structured interconnection operators. }
We denote this interconnection between systems $G$ and $N_E$ via $M$ and $M^{\top}$ as $G \#_M N_{E}$. 
This structure will be the basis for the passivity and stability analysis in the sequel.

\section{Decentralized Stability Certificate} \label{sec:certificate}
We now present and prove the main result of this paper. 

\subsection{Main Result}

\begin{theorem}[Decentralized Stability Certificate]\label{thm:main_result}
Consider the device and network models in \eqref{eq:ibr_model}, \eqref{eq:line_model}, and \eqref{eq:lifting_and_aggregation}, interconnected as $G(s) \#_M N_E(s)$ in Fig.~\ref{fig:our_original_loop}.
Assume that the operating point satisfies $\cos\theta_{ij}>0$ for all $\{i,j\}\in E$. 
If the following condition holds:
For each line $e=\{i,j\}\in E$,
    \begin{equation}\label{eq:condition_line_1}
        \begin{aligned}
            \left(2\frac{d_{e}}{b_{e}|v_i||v_j|}-\frac{2}{\cos\theta_{ij}}\right)
            \sqrt{\left(\frac{Q_{ij}-Q_{ji}}{b_{e}|v_i||v_j|}\right)^2+4} \\
            +\left(\frac{d_{e}}{b_{e}|v_i||v_j|}\right)^2
            -\frac{2}{\cos\theta_{ij}}\frac{d_{e}}{b_{e}|v_i||v_j|}
            +4 \ge 0, 
        \end{aligned}
    \end{equation}
    where $\sum_{e:i\in e} d_e = \frac{1}{k_{q,i}}$. 
    Then $G(s) \#_M N_E(s)$ is internally stable.
\end{theorem}

\begin{remark}\label{rem:kp}
    \textcolor{black}{This certificate imposes no restriction on $k_{p,i}$ or $\tau_{\omega,i}$, in agreement with Remark~5.11 of \cite{schiffer2014conditions}, where stability of the lossless microgrid is likewise independent of the frequency droop gains and filter time constants.} 
\end{remark}

\subsection{Proof of Main Theorem} \label{sec:proof}

\begin{figure}
    \centering
    \includegraphics[width=0.7\linewidth]{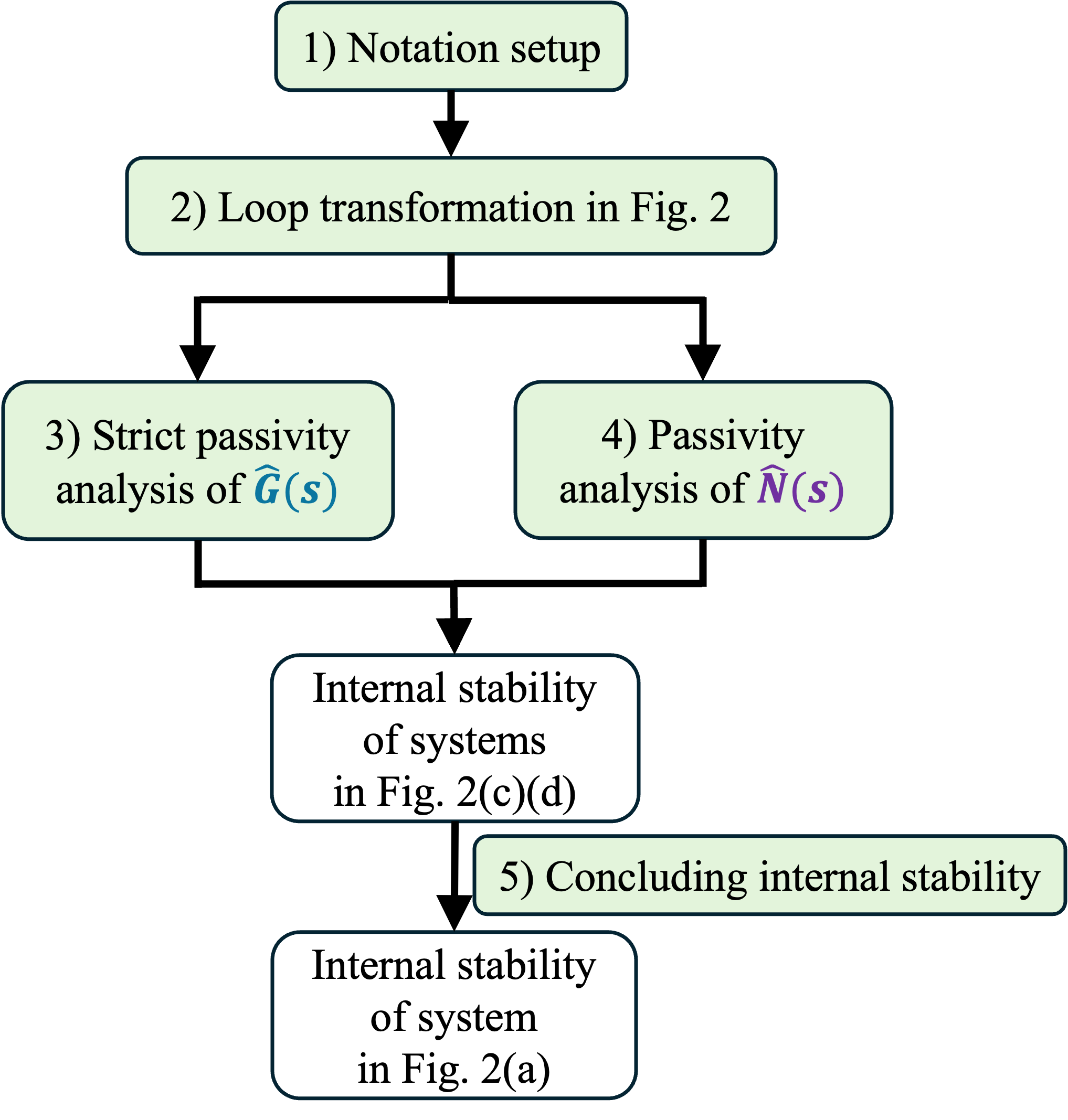}
    \caption{Flow chart of the proof of Theorem \ref{thm:main_result}. The colored blocks represent the steps. }
    \label{fig:proof_flow_chart}
\end{figure}

We prove Theorem \ref{thm:main_result} by following the flow chart in Fig.~\ref{fig:proof_flow_chart}. 

\subsubsection{Notation setup}
To facilitate the stability analysis, we normalize the system parameters around the operating point. 
For each line \(e=\{i,j\}\), define $\theta_e:=\theta_{ij}$, $c_{\theta_e}:=\cos\theta_e$, $s_{\theta_e}:=\sin\theta_e$. 
Furthermore, inspired by \cite{pates2014scalable}, we introduce the normalized gain and the corresponding diagonal matrix
\begin{equation}\label{eq:dbar_def}
\bar d_{e} := \frac{d_{e}}{b_{e}|v_i||v_j|}, \quad \bar D_{e} = \mathrm{diag}(0,\bar d_{e}, 0,\bar d_{e}),
\end{equation}
to be used in our loop transformation, and the normalized reactive power flow $p_{ij}$, $q_{ij}$ and $q_{ji}$
\[
p_{ij}:=\frac{P_{ij}}{b_{e}|v_i||v_j|},\qquad
q_{ij}:=\frac{Q_{ij}}{b_{e}|v_i||v_j|},\qquad
q_{ji}:=\frac{Q_{ji}}{b_{e}|v_i||v_j|}.
\]
Also notice the standard power-flow identities of lossless lines 
\begin{equation}\label{eq:lossless_pf_identities_norm}
p_{ij}=\sin\theta_{ij},
\qquad
q_{ij}-q_{ji}=\frac{Q_{ij}-Q_{ji}}{b|v_i||v_j|}
=\frac{|v_i|}{|v_j|}-\frac{|v_j|}{|v_i|}. 
\end{equation}

We will prove our results using the following coordinates
\begin{align}
\Sigma :=& \frac{|v_i|}{|v_j|}+\frac{|v_j|}{|v_i|}+\bar d_{e}=\sqrt{(q_{ij}-q_{ji})^2+4}+\bar d_{e}, \label{eq:Sigma_def}\\
\Delta :=& \frac{|v_i|}{|v_j|}-\frac{|v_j|}{|v_i|}=q_{ij}-q_{ji}.\label{eq:Delta_def}
\end{align}

\subsubsection{Loop transformation}

The map in \eqref{eq:line_model} is in general not passive, as $J_E$ is in general not positive semidefinite. 
To facilitate the subsequent analysis, we consider the systems in Fig.~\ref{fig:our_loop_transformation}, in which Fig.~\ref{fig:our_original_loop} is the original system, and can be equivalently expressed as in Fig.~\ref{fig:first_step_loop_shifting}, where  
\begin{equation}
    G_0(s) =  \mathrm{blkdiag}\left(\begin{bmatrix}
        \frac{k_{p,i}}{s(\tau_{\omega,i}s+1)} & 0\\
0 & \frac{k_{q,i}}{\tau_{v,i}s+1}
    \end{bmatrix}\right), 
\end{equation}
with the dynamics $\begin{bmatrix}
        \Delta \theta \\
        \Delta \ln |v|
    \end{bmatrix}= G_{0}(s) \begin{bmatrix}
        -\Delta P \\
        -\Delta Q
    \end{bmatrix}$. 

In Fig.~\ref{fig:before_loop_shifting}, $G'(s):=sG_{0}(s)$ characterizes the dynamics
\begin{equation}\label{eq:G_prime_dyn}
    \begin{bmatrix}
        \Delta \omega \\
        \Delta \dot{\ln |v|}
    \end{bmatrix} = G'(s) \begin{bmatrix}
        -\Delta P \\
        -\Delta Q
    \end{bmatrix}. 
\end{equation}

For each $e=\{i,j\}$, we choose loop shifting parameter $d_{e} \ge 0$, and apply the corresponding loop shifting matrix $D_{E}$ to the line dynamics in Fig.~\ref{fig:before_loop_shifting}, which yields the loop shifted interconnection in Fig.~\ref{fig:our_loop_shifted_loop_color}. 
On the IBR side, this loop transformation induces the additional nodal positive feedback $\dfrac{1}{s}K\begin{bmatrix}
    \Delta \omega \\
    \Delta \dot{\ln |v|}
\end{bmatrix}$, where $K:=M\,D_E\,M^\top$. 
By construction, $K$ is diagonal and acts only on the voltage-magnitude channels: $K=\mathrm{diag}(0,\kappa_1,0,\kappa_2,\dots,0,\kappa_n)$, where $\kappa_i = \sum_{e:i \in e} d_e$. 
Note that $D_e$ acts locally on each line, whereas $K$ aggregates these effects at the nodal level. 

\subsubsection{Strict passivity analysis of loop shifted devices $\hat G(s)$}
The diagonal term $\frac{1}{s}K$ is equivalent to introducing, at each bus $i$, a local positive feedback in the $\dot{\ln|v_i|}$ channel with gain $\kappa_i$. 

Given the IBR dynamics in \eqref{eq:G_prime_dyn}, the additional feedback $\kappa_i/s$ acts as a positive feedback around $\frac{k_{q,i}s}{\tau_{v,i}s+1}$ in voltage channel. 
Consequently, the effective voltage-channel dynamics become $\frac{k_{q,i}s}{\tau_{v,i}s+1-\kappa_i k_{q,i}}$, which remains \emph{strictly} passive with $\kappa_i = \frac{1}{k_{q,i}}$. 
\subsubsection{Passivity analysis of loop shifted lines $\hat N(s)$}
With the loop transformation introduced, each of the line transfer function $\frac{1}{s}J_e$ can be replaced by the shifted map 
\begin{equation}\label{eq:shifted_line_map}
\frac{1}{s}J_e
\;=\;
\frac{1}{s}(J_e + D_{e})
\;-\;
\frac{1}{s}D_{e}. 
\end{equation}
Equivalently, we define the \emph{passivized} line block
\begin{equation}\label{eq:Nline_bar_def}
\bar N_{\mathrm{ln},e}(s):=\frac{1}{s}\,(J_e + D_{e}),
\qquad
N_{D,e}(s):=\frac{1}{s}\,D_{e},
\end{equation}
so that the original line is decomposed into a passive candidate $\bar N_{\mathrm{ln},e}$ and a diagonal integrator $N_{D,e}$. 
Since $\frac{1}{s}$ is positive real, a sufficient condition for $\bar N_{\mathrm{ln},e}(s)$ to be passive is $\bar J_e + \bar D_{e}\succeq 0$, and is a simple semidefinite condition on $d_{e}$. 

Stacking over all lines gives the passivized stacked line map
\[
\bar N_{\mathrm{ln}}(s)=\frac{1}{s}\,(J_E+D_E),
\qquad
N_D(s)=\frac{1}{s}\,D_E.
\]

For each line, define $H_e:= \bar{J}_e\!+\!\bar{D}_e$, giving
\begin{equation}\label{eq:He_explicit}
    H_e\!=\!\begin{bmatrix}
c_{\theta_e} \!&\! s_{\theta_e} \!&\! \!-c_{\theta_e} \!&\! s_{\theta_e} \\
s_{\theta_e} \!&\! 2\frac{|v_i|}{|v_j|}\!-\!c_{\theta_e}\!+\!\bar d_{e} \!&\! \!-s_{\theta_e} \!&\! \!-c_{\theta_e} \\
\!-c_{\theta_e} \!&\! \!-s_{\theta_e} \!&\! c_{\theta_e} \!&\! \!-s_{\theta_e} \\
s_{\theta_e} \!&\! \!-c_{\theta_e} \!&\! \!-s_{\theta_e} \!&\! 2\frac{|v_j|}{|v_i|}\!-\!c_{\theta_e}\!+\!\bar d_{e}
\end{bmatrix}. 
\end{equation}

The definition in \eqref{eq:Sigma_def}--\eqref{eq:Delta_def} can be equivalently written as
\[
2\frac{|v_i|}{|v_j|}+\bar d_{e} = \Sigma+\Delta,\qquad
2\frac{|v_j|}{|v_i|}+\bar d_{e} = \Sigma-\Delta.
\]

To analyze the definiteness of $H_e$, we apply a congruence transformation $\tilde H_e := T^\top H_e T$ that transforms the basis into $\Sigma$ and $\Delta$, thereby block-diagonalizing the matrix and simplifying the stability check. 
We pick
\[T:=\begin{bmatrix}
1&0&1&0\\
0&1&0&1\\
1&0&-1&0\\
0&1&0&-1
\end{bmatrix}. \]
Applying the transformation to \eqref{eq:He_explicit} and simplifying yields
\begin{equation}\label{eq:THT_form}
\tilde H_e 
=\begin{bmatrix}
0 & 0 & 0 & 0\\
0 & 2\Sigma-4c_{\theta_e} & 4s_{\theta_e} & 2\Delta \\
0 & 4s_{\theta_e} & 4c_{\theta_e} & 0 \\
0 & 2\Delta & 0 & 2\Sigma
\end{bmatrix}, 
\end{equation}
thus $H_e\succeq 0 \Leftrightarrow G_e\succeq 0 \Leftrightarrow \hat G_e:=\frac12 G_e\succeq 0$, where
\[
G_e\!:=\!
\begin{bmatrix}
2\Sigma-4c_{\theta_e} & 4s_{\theta_e} & 2\Delta \\
4s_{\theta_e} & 4c_{\theta_e} & 0 \\
2\Delta & 0 & 2\Sigma
\end{bmatrix}\!, \hat G_e
\!=\!
\begin{bmatrix}
\Sigma-2c_{\theta_e} \!&\! 2s_{\theta_e} \!&\! \Delta \\
2s_{\theta_e} \!&\! 2c_{\theta_e} \!&\! 0 \\
\Delta \!&\! 0 \!&\! \Sigma
\end{bmatrix}
\!\succeq \! 0. 
\]

We assume $c_{\theta_e}>0$. 
Then $\hat G_e\succeq 0$ is equivalent to the Schur complement of the middle block being PSD:
\[
S_e
\!\!:=\!\!
\begin{bmatrix}
\Sigma-2c_{\theta_e} \!&\! \Delta\\
\Delta \!&\! \Sigma
\end{bmatrix}
\!-\frac{1}{2c_{\theta_e}}\!
\begin{bmatrix}
2s_{\theta_e}\\ 0
\end{bmatrix}\!\!
\begin{bmatrix}
2s_{\theta_e} \!\!&\!\! 0
\end{bmatrix}\!=\!\begin{bmatrix}
\Sigma\!-\!2c_{\theta_e}\!\!-\!\dfrac{2s_{\theta_e}^2}{c_{\theta_e}} \!\!&\!\! \Delta\\
\Delta \!\!&\!\! \Sigma
\end{bmatrix},
\]
it is PSD if and only if its diagonals and determinant are nonnegative. 

We substitute $\Sigma$ and $\Delta$ into \eqref{eq:condition_line_1}, thus the condition ensures 
\begin{equation} \label{eq:psd_condition_se}
    -\Delta^2 \!+\! \Sigma^2 \!-\! \dfrac{2 \Sigma}{c_{\theta_e}}
        \!=\! \Sigma^2\!-\!2c_{\theta_e}\Sigma\!-\!\dfrac{2s_{\theta_e}^2}{c_{\theta_e}}\Sigma\!-\!\Delta^2 \!=\! \mathrm{det}(S_e) \geq 0. 
\end{equation}

Moreover, \eqref{eq:psd_condition_se} guarantees $\Sigma^2 - \frac{2}{c_{\theta_e}}\Sigma \ge \Delta^2 \ge 0$, i.e., $\Sigma\!\left(\Sigma-\frac{2}{c_{\theta_e}}\right)\ge 0$. 
\vspace{-1pt}
Since $\Sigma \ge 0$, this implies $\Sigma \ge \frac{2}{c_{\theta_e}}$, i.e., \eqref{eq:psd_condition_se} implies the smallest diagonal element in $S_e$ is nonnegative. 
Thus the condition ensures $S_e$ is PSD, implying $\hat G_e$, $G_e$, $H_e$ are PSD, further implying passivity of loop-shifted lines. 

\subsubsection{Concluding internal stability}

Given that the condition ensures the strict passivity of the loop shifted IBRs and the passivity of the loop shifted lines, the loop-shifted system $\hat G(s)\#_M \hat N(s)$ in Fig.~\ref{fig:our_loop_shifted_loop_color} is internally stable. 
Since $sG_0(s)$ is stable and there is no RHP pole-zero cancellation between $sG_0(s)$ and $\frac{1}{s}J_{E}$, we apply Lemma~\ref{lem:loop_transformation_equiv} and conclude that system $sG_0(s)\#_M \frac{1}{s}J_{E}$ is also internally stable, i.e., signals $(\Delta P, \Delta Q, \Delta \omega, \Delta \dot{\ln |v|})\to 0$ as $t\to \infty$. 
To further show that $\Delta \ln |v|\to 0$ as $t\to \infty$, which can be implied by the internal stability of the original system $G(s)\#_M N_E(s)$ in Fig.~\ref{fig:our_original_loop}, we adapt the proof method from \cite{haberle2025decentralized} and state the following. 
\begin{lemma}\label{lemma:fvt}
    Consider system $G(s)\#_M N_E(s)$ in Fig.~\ref{fig:our_original_loop}, and system $\hat G(s) \#_M \hat N(s)$ in Fig.~\ref{fig:our_loop_shifted_loop_color}. 
    Then, internal feedback stability of the minimal realization of all four closed-loop transfer functions of the loop-shifted system $\hat G(s)\#_M\hat N(s)$ implies internal feedback stability of $G(s)\#_M N_{E}(s)$.
\end{lemma}
\begin{proof}
    See Appendix \ref{appendix:fvt}. 
\end{proof}


\section{Numerical Example} \label{sec:num_eg}
We consider a two-bus system composed of two GFM inverters interconnected by a single transmission line. 
\rac{To validate the proposed stability certificates, we adopt a third-order nonlinear GFM model, as detailed in Appendix~\ref{appendix:gfm}, whose linearization yields \eqref{eq:ibr_model}. 
We begin by linearizing the nonlinear dynamics around a stable equilibrium point obtained with the following parameter settings: $p^{*}_1 = -4$ p.u., $p^{*}_2 = 4$ p.u., $q^{*}_1 = 1.2$ p.u., $q^{*}_2 = 0.5$ p.u., $k_{p,1} = k_{p,2} = 1.8$ p.u., $k_{q,1} = k_{q,2} = 0.3$ p.u., and $\omega_b = 120\pi$ rad/s. 
These parameters are selected to represent an operating condition characterized by voltage mismatch and line loading, as summarized below.}
\begin{itemize}
\item Grid State: $|v_1| = 1.0218$ p.u., $|v_2| = 0.9919$ p.u., and phase angle $\theta_{1} = 0^\circ$, $\theta_{2} = 23.2448^\circ$.
\item Line Parameter: $b_{e} = 10.0$ p.u.
\item Control Gains: $k_{q,1} = k_{q,2} = 0.3$. 
\end{itemize}
From the above, $\cos\theta = 0.9188$, $\frac{Q_{ij}-Q_{ji}}{b_e|v_1||v_2|} = 0.0593$. 
One can verify that, with $\kappa_i = \sum_{e:i \in e} d_e = \frac{1}{k_{q,i}}$, condition is satisfied. 

\begin{figure}
    \centering
    \includegraphics[width=0.75\linewidth]{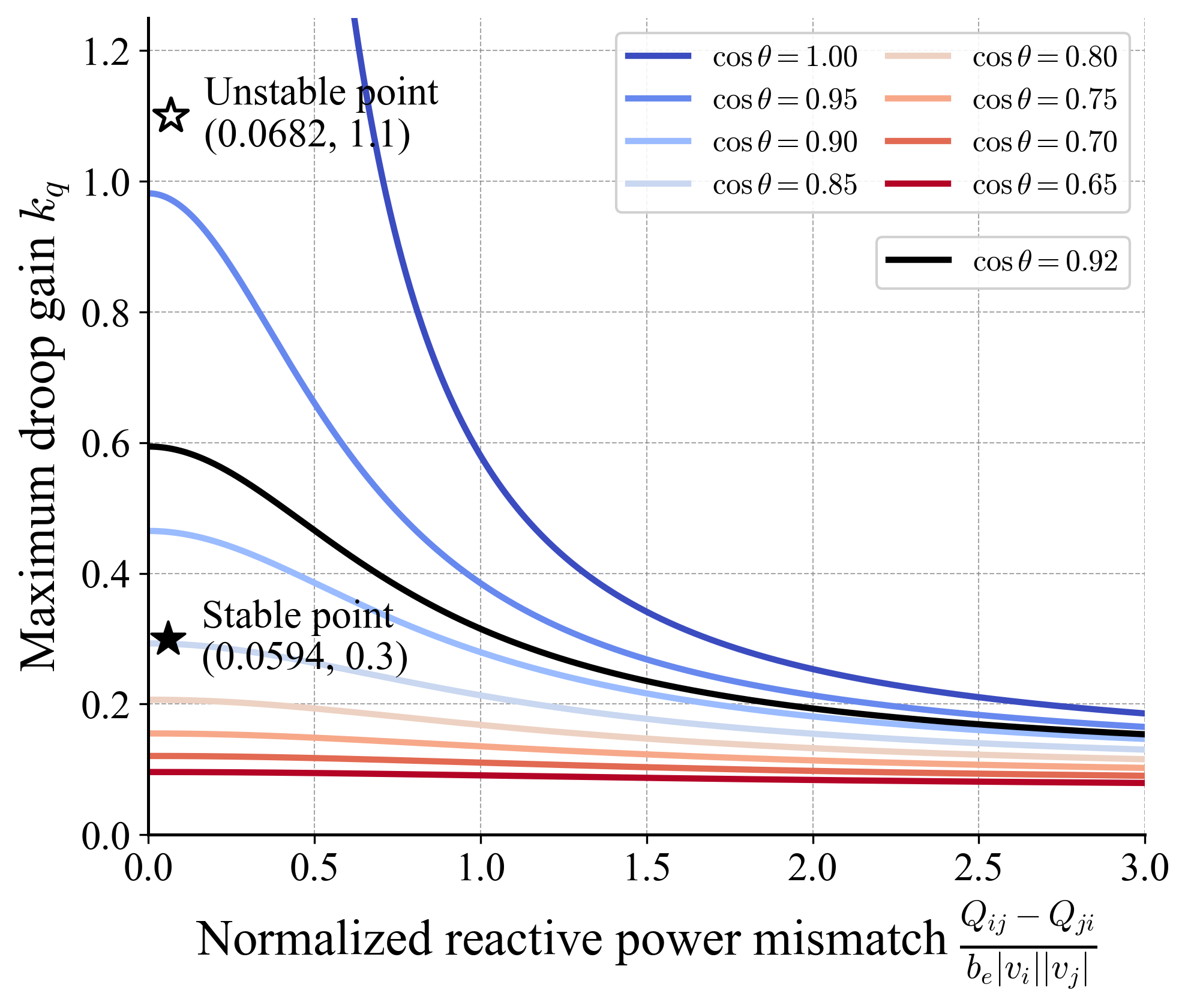}
    \caption{$\cos\theta$ level sets illustrating how maximum $k_q$ varies with $\frac{Q_{ij}-Q_{ji}}{b_e|v_i||v_j|}$. \textcolor{black}{The region below each curve is certified stable.} }
    \label{fig:mismatch_vs_max_gain}
\end{figure}

\begin{figure}
    \centering
    \includegraphics[width=1\linewidth]{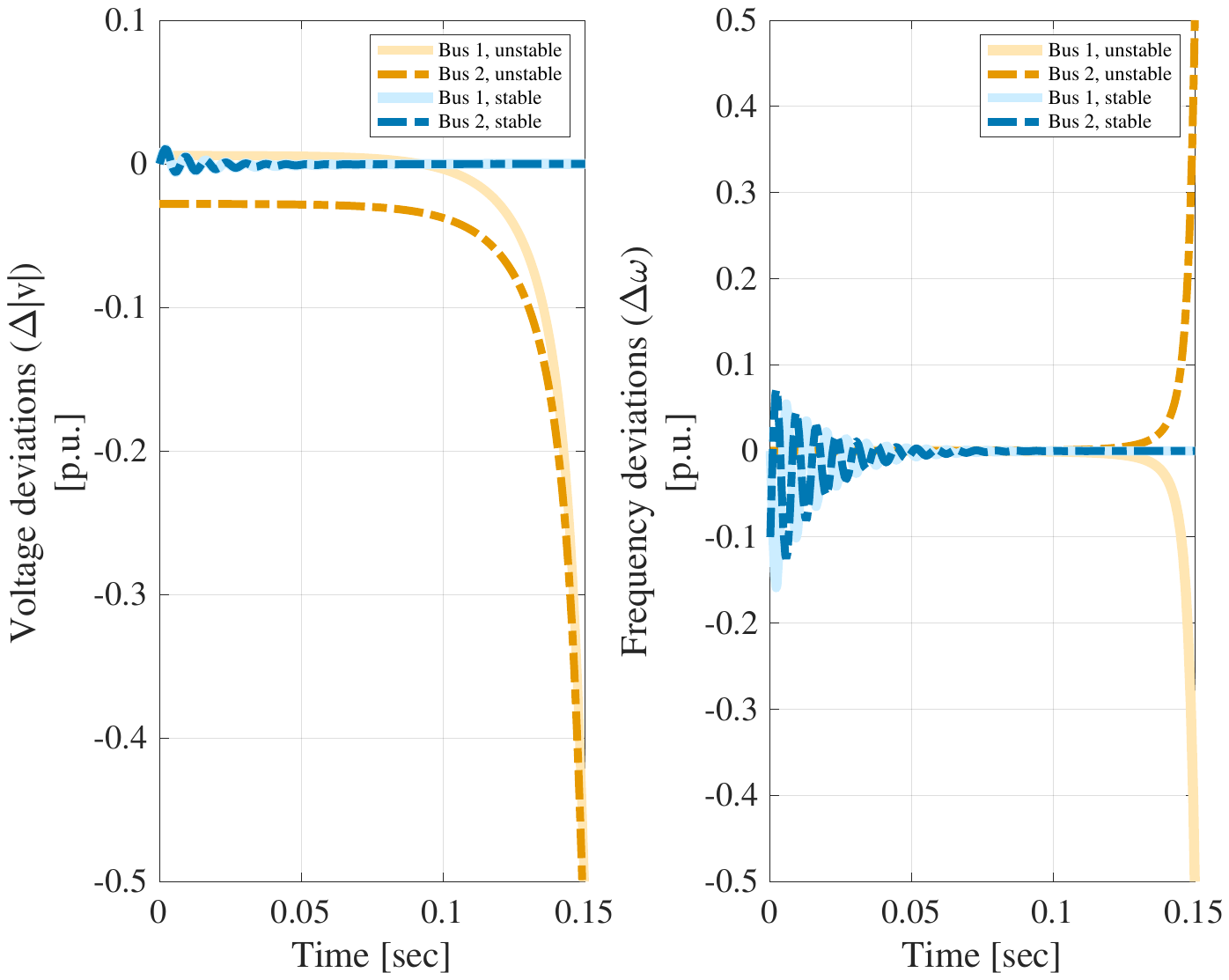}
    \caption{Dynamic response of the 2-bus system considering one stable and one unstable scenario.}
    \label{fig:two-bus-time}
\end{figure}

In Fig.~\ref{fig:mismatch_vs_max_gain}, we plot the maximum admissible droop gain to maintain stability versus the $\frac{Q_{ij}-Q_{ji}}{b_{e} |v_i| |v_j|}$. 
The region below each curve represents the stable operating set for a given $\cos\theta$. 
The varying curves represent different $\cos\theta$ values. 
We observe a clear inverse relationship between the reactive power mismatch and the stability margin. 
As the mismatch increases, 
the maximum allowable droop gain decreases monotonically. 
Furthermore, the system stability margin degrades significantly with smaller $\cos\theta$.
This indicates that operating points with high reactive loading and large voltage disparity require much more conservative controller tuning to ensure stability. 
To validate our numerical example, in Fig.~\ref{fig:mismatch_vs_max_gain}, we highlight the stability boundary for $\cos\theta = 0.92$ and the operating point from the example with a ``$\filledstar$''. 
Since the operating point lies within the safe region, the system is confirmed to be stable. 

\rac{Furthermore, we examine an alternative operating point obtained by setting $k_{q,1} = 1.1$ p.u., and $k_{q,2} = 1$ p.u., while keeping all other parameters unchanged.}
This operating point is marked with a ``$\hollowstar$'' in Fig.~\ref{fig:mismatch_vs_max_gain}, and violates the stability conditions established in Theorem~\ref{thm:main_result}, thereby demonstrating the ability of the proposed certificates to correctly characterize its stability properties. 
Fig.~\ref{fig:two-bus-time} presents time-domain simulations that corroborate the theoretical analysis for both the stable and unstable operating points. 
The two panels depict the voltage deviations and the frequency synchronization dynamics, respectively. 
This confirms the sufficiency of the proposed certificate to guarantee dissipation of oscillation energy.

\section{Conclusion} \label{sec:conclusion}
In this paper, we present a decentralized method for certifying the small-signal stability of IBR-dominated grids, with a specific focus on the impact of reactive power flows. 
We demonstrate that the network's operating point is an active factor in determining local passivity. 
By passivizing both devices and lines using loop transformation, we derive stability criteria that are both rigorous and practically implementable. 
These findings offer a clearer path towards ``plug-and-play'' integration of IBRs, ensuring robust operation even as grid complexity and interactions increase.

\section{AI Usage Disclosure}

The authors used AI tools to assist with narrative polishing, including spell-checking, visualization coding and debugging. None of the narrative was directly produced by AI. All models and ideas are the sole intellectual property of the authors. 

\bibliographystyle{ieeetr}
\bibliography{references}

\appendix

\subsection{Proof of Theorem \ref{thm:passivity_stability}}\label{appendix:thm1}
The internal stability of the feedback system in Fig.~\ref{fig:before_loop_transformation_012026} is characterized by the following:
\begin{equation}\label{eq:gang_of_four_app}
    \begin{bmatrix}
        (I+H_2 H_1)^{-1} & (I+H_2 H_1)^{-1}H_2 \\
        H_1(I+H_2 H_1)^{-1} & H_1(I+H_2 H_1)^{-1}H_2
    \end{bmatrix}. 
\end{equation}

Given the strict passivity of $H_2$, it suffices to only check the stability of $H_1(I+H_2 H_1)^{-1}$ \cite{zhou1996robust}, which is equivalent to $-1\notin \sigma(H_2 H_1)$.\footnote{Here $\sigma(A)$ represents the set of eigenvalues of $A$.} 
Suppose $-1\in \sigma(H_2 H_1)$, let $v\neq \vec{0}$ be the corresponding eigenvector. 
Thus $-v=H_2 H_1 v$. 
It follows that
\begin{equation}\label{eq:pass_proof_1}
    -v^* H_1^* v = v^* H_1^* H_2 H_1 v. 
\end{equation}
We take the conjugate transpose of both sides, it yields
\begin{equation}\label{eq:pass_proof_2}
    -v^* H_1 v = v^* H_1^* H_2^* H_1 v. 
\end{equation}
Summing both sides of \eqref{eq:pass_proof_1} and \eqref{eq:pass_proof_2} gives
\begin{equation}\label{eq:pass_proof_3}
    -v^* (H_1 + H_1^*) v = v^* H_1^* (H_2+H_2^*) H_1 v. 
\end{equation}

By the passivity of $H_1$ and strict passivity of $H_2$, the left hand side of \eqref{eq:pass_proof_3} is nonpositive, while the right hand side is strictly positive for $v\neq \vec{0}$. 
Thus there is a contradiction, indicating that $-1\notin \sigma(H_2 H_1)$ for every finite frequency under consideration. 

Moreover, since $H_1$ and $H_2$ are passive transfer matrices, they are proper, and hence the limits 
\begin{equation}
    H_1(j\infty)
    :=
    \lim_{\omega\to\infty} H_1(j\omega),
    \quad
    H_2(j\infty)
    :=
    \lim_{\omega\to\infty} H_2(j\omega)
\end{equation}
exist and are finite.
Let $\lambda\in\sigma\!\left(H_2(j\infty)H_1(j\infty)\right)$. 
Then by condition ii, 
\begin{equation}
    |\lambda|\leq \bar{\sigma}\!\left(H_2(j\infty)H_1(j\infty)\right) \leq \bar{\sigma}\!\left(H_2(j\infty)\right) \bar{\sigma}\!\left(H_1(j\infty)\right) < 1. 
\end{equation}
Therefore every eigenvalue of $H_2(j\infty)H_1(j\infty)$ has magnitude strictly less than one. 
In particular, $-1\notin\sigma\!\left(H_2(j\infty)H_1(j\infty)\right)$. 

Therefore, the interconnection is internally stable given the conditions in Theorem \ref{thm:passivity_stability}. 
\subsection{Proof of Lemma~\ref{lem:loop_transformation_equiv}} \label{appendix:cor21}
Define 
\begin{equation}
    T(s) \!:=\! \bigl(I \!+\! H_2(s)H_1(s)\bigr)^{-1}\!H_2(s), 
\end{equation}
\begin{equation}
    T'(s) \!:=\! \bigl(I \!+\! H_2'(s)H_1'(s)\bigr)^{-1}\!H_2'(s). 
\end{equation}
First suppose $H_1(s) \# H_2(s)$ is internally stable. 
Then $T(s)$ is stable. 
We show $T'(s)$ is also stable. 
\[
\begin{aligned}
T'(s)
=& \bigl(I - \Gamma(s)H_1(s)\bigr)\bigl(I + H_2(s)H_1(s)\bigr)^{-1}\bigl(H_2(s) + \Gamma(s)\bigr) \\
=& \bigl(I - \Gamma(s)H_1(s)\bigr)\bigl(T(s) + \Gamma(s) - T(s)H_1(s)\Gamma(s)\bigr).
\end{aligned}
\]
Since $H_1(s)$, $T(s)$, $\Gamma(s)$ are stable, the right-hand side is a finite sum and product of stable transfer matrices. 
Hence $T'(s)$ is stable. 
Since $H_1'(s)$ is stable, Lemma~\ref{lemma:zhou_55} implies that $H_1'(s) \# H_2'(s)$ is internally stable.

Conversely, suppose that $H_1'(s) \# H_2'(s)$ is internally stable. Then $T'(s)$ is stable. The transformation is invertible, with
\[
H_1(s) = H_1'(s)\bigl(I + \Gamma(s)H_1'(s)\bigr)^{-1}, \quad H_2(s) = H_2'(s) - \Gamma(s).
\]
Applying the same argument with $\Gamma(s)$ replaced by $-\Gamma(s)$, 
\[
T(s) = \left(I + \Gamma(s)H_1'(s)\right)\left(T'(s) - \Gamma(s) + T'(s)H_1'(s)\Gamma(s)\right).
\]
Since $H_1'(s)$, $T'(s)$, $\Gamma(s)$ are stable, $T(s)$ is stable. Since $H_1(s)$ is stable, Lemma~\ref{lemma:zhou_55} implies that $H_1(s) \# H_2(s)$ is internally stable.

\subsection{Proof of Lemma \ref{lemma:fvt}}\label{appendix:fvt}
Since $sG_0(s)$ is stable and there is no RHP pole-zero cancellation between $sG_0(s)$ and $\frac{1}{s}J_{E}$, we apply Lemma~\ref{lem:loop_transformation_equiv} and conclude that system $sG_0(s)\#_M \frac{1}{s}J_{E}$ is also internally stable. 
Moreover, notice that 
\begin{equation}
    \left(I+G'\frac{1}{s}J_{E}\right)^{-1}G' = \mathrm{diag}(1,s,\cdots,1,s)\cdot \left(I+GN_{E}\right)^{-1}G, 
\end{equation}
thus we can conclude the stability of $\left(I+GN_{E}\right)^{-1}G$ if the step response of the voltage derivatives converges to zero, i.e.,
\begin{equation}\label{eq:limit}
    \lim_{s\to 0} \!s\frac{1}{s}\left(\!I\!+\!G'\!\frac{1}{s}J_{E}\!\right)^{\!-1}\!\!\!G' = I\otimes \begin{bmatrix}
        * \!&\! 0 \\
        0 \!&\! 0
    \end{bmatrix}. 
\end{equation}

To show \eqref{eq:limit}, define permutation matrix $\mathcal{P} = \begin{bmatrix} \mathcal{P}_{u}^\top & \mathcal{P}_{l}^\top \end{bmatrix}^\top$: 
\begin{equation*}
    (\mathcal{P}_{u})_{i,j} = \begin{cases} 1 & \text{if } j = 2i-1 \\ 0 & \text{otherwise} \end{cases}, \quad (\mathcal{P}_{l})_{i,j} = \begin{cases} 1 & \text{if } j = 2i \\ 0 & \text{otherwise} \end{cases}. 
\end{equation*}

We apply the permutation matrix $\mathcal{P}$ to the limit, 
\begin{equation}\label{eq:limit_permutated}
    \lim_{s\to 0} \mathcal{P}\left(I+G'\frac{1}{s}J_{E}\right)^{-1}G'\mathcal{P}^{\top} = \mathcal{P}I\otimes \begin{bmatrix}
        * & 0 \\
        0 & 0
    \end{bmatrix}\mathcal{P}^{\top}=\begin{bmatrix}
        * & 0 \\
        0 & 0
    \end{bmatrix}. 
\end{equation}

Due to orthogonality of $\mathcal{P}$, the LHS of \eqref{eq:limit_permutated} is equivalent to 
\begin{equation}\label{eq:limit_pi}
    \lim_{s\to 0} \left(I+G^{\pi}N^{\pi}\right)^{-1}G^{\pi}, 
\end{equation}
where 
\begin{equation}
    G^{\pi} \!\!=\!\! \mathcal{P} G' \mathcal{P}^{\!\top}\!\!\!=\!\!\!\begin{bmatrix}
        \mathrm{diag}\!\left(\!\frac{k_{p,i}}{\tau_{\omega,i}s+1} \!\right) \!&\! 0 \\
        0 \!&\! \mathrm{diag}\!\left(\!\frac{k_{q,i}s}{\tau_{v,i}s+1} \!\right)
    \end{bmatrix}\!\!, 
    N^{\pi} \!\!=\!\! \frac{1}{s}\mathcal{P} J_{E} \mathcal{P}^{\top}\!\!. 
\end{equation}

Notice that $N^{\pi}$ contains a $1/s$ factor. 
Define $K^{\pi} = s N^{\pi}$ and take the limits:
\begin{equation}
    \lim_{s\to 0}G^{\pi} \!\!=\!\!\begin{bmatrix}
        \mathrm{diag}(k_{p,i}) \!&\! 0 \\
        0 \!&\! 0
    \end{bmatrix} \!:=\! \bar{G}, \,
    \lim_{s\to 0} s N^{\pi} \!\!=\!\! \begin{bmatrix}
        K_{11} \!&\! K_{12} \\
        K_{21} \!&\! K_{22}
    \end{bmatrix} \!:=\! K^{\pi}. 
\end{equation}
Thus, substituting $(I + G^\pi N^\pi)^{-1} = s(sI + G^\pi K^\pi)^{-1}$, \eqref{eq:limit_pi} becomes 
\begin{equation}
\begin{aligned}
    &\! \lim_{s \to 0} s \left(sI+\bar{G} K^\pi \right)^{-1}\bar{G} \\
    =&\! \lim_{s \to 0} s\!\begin{bmatrix}
        \!\bigl(\!sI\!+\!\mathrm{diag}(k_{p,i})K_{11}\!\bigr)^{-1} \!\!&\!\! \mathrm{diag}(k_{p,i})\!K_{12} \\
        0 \!\!&\!\! sI
    \end{bmatrix}^{-1}\!\!\begin{bmatrix}
        \mathrm{diag}(k_{p,i}) \!\!&\!\! 0 \\
        0 \!\!&\!\! 0
    \end{bmatrix}\\
    =&\! \lim_{s \to 0} s\!\begin{bmatrix}
        \!\bigl(sI+\mathrm{diag}(k_{p,i})K_{11}\bigr)^{-1} & \star \\
        0 & \frac{1}{s}I
    \end{bmatrix}\begin{bmatrix}
        \mathrm{diag}(k_{p,i}) \!\!&\!\! 0 \\
        0 \!\!&\!\! 0
    \end{bmatrix}\\
    =&\! \lim_{s \to 0} \begin{bmatrix}
        \!s\bigl(sI+\mathrm{diag}(k_{p,i})K_{11}\bigr)^{-1}\mathrm{diag}(k_{p,i}) & 0 \\
        0 & 0
    \end{bmatrix}
\end{aligned}
\end{equation}
which yields the required structure.

\subsection{GFM Modeling Equations}\label{appendix:gfm}
The nonlinear differential-algebraic equations describing the GFM dynamics are given below in per units. 
\subsubsection{Differential Equations}
\begin{align}
    \dot{\theta} &\ =\ \omega_b \delta \omega \label{eqn:sm1}
        \\
        \tau_{\omega}\dot{\delta \omega} &\ =\ -\delta \omega + (p^{*}-p)k_p \label{eqn:sm2}
        \\
        \tau_{v}\dot{\delta |v|} &\ =\ -\delta |v| + (q^{*}-q)k_q \label{eqn:sm3}
\end{align}
\subsubsection{Algebraic Equations}
\begin{align}
    \delta\omega &\ =\ \omega-\omega_0 \label{eqn:am1}
    \\
    \delta |v| &\ =\ |v|-|v_0| \label{eqn:am2}
    \\
    p &\ =\ v_\text{D}i_\text{D} + v_\text{Q}i_\text{Q} \label{eqn:am3}
    \\
    q &\ =\ v_\text{Q}i_\text{D} - v_\text{D}i_\text{Q} \label{eqn:am4}
    \\
    |v| &\ =\ \sqrt{v_\text{D}^2 + v_\text{Q}^2} \label{eqn:am5}
    \\
    v_\text{D}\sin\theta &\ =\ v_\text{Q}\cos\theta \label{eqn:am6}
\end{align}
Here, $\tau_{\omega}$ \& $\tau_{v}$ represent frequency \& voltage filter time constants, $k_{p}$ \& $k_{q}$ denote the frequency \& voltage droop gains, respectively. 
$|v_0|$, $\omega_0$ \& $\omega_b$ denote the nominal voltage, angular frequency setpoint \& base frequency of the inverter. $v_\text{D}$ \& $v_\text{Q}$ refers to global D- and Q-axis output voltage at PCC, while $i_\text{D}$ \& $i_\text{Q}$ refers to global D- and Q-axis output current at PCC, respectively. $\theta_c$ denotes the PCC angle w.r.t a global DQ-frame. 
The reference frames are illustrated in Fig.~\ref{fig:gfm_ph}. 
\begin{figure}[ht!]
    \centerline{\includegraphics[scale=0.18]{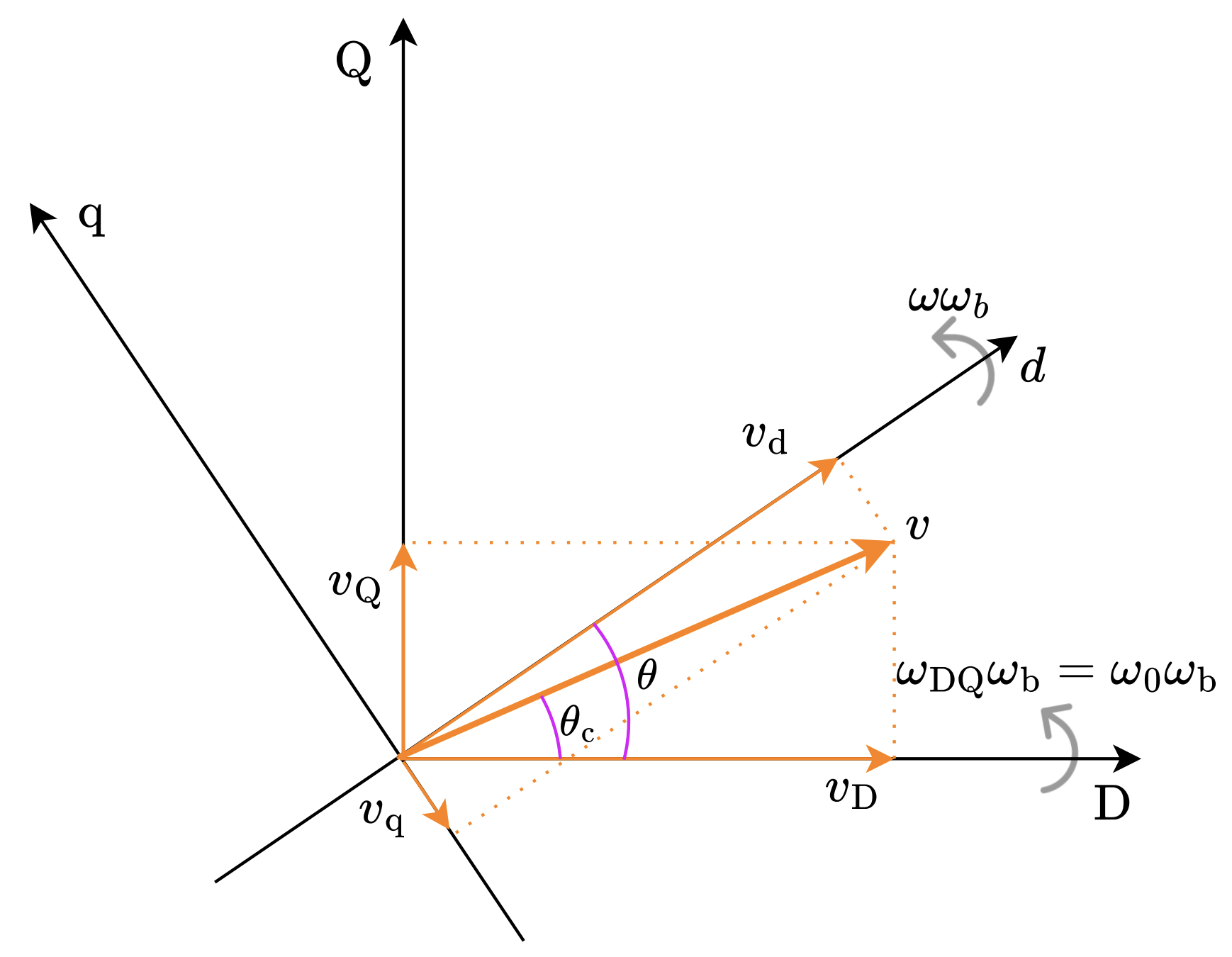}} 
    \caption{Phasor diagram and reference frames of the GFM IBR.}
    \label{fig:gfm_ph} 
\end{figure} 

\balance

\endgroup
\end{document}